\documentclass[a4paper]{scrartcl}

\usepackage{amsmath}

\usepackage{amsfonts}
\usepackage{amssymb}
\usepackage{amsthm}
\usepackage{hyperref}

\usepackage[english]{babel}
\usepackage{cite}
\usepackage{todonotes}

\usepackage{comment}
\def\IsExtended{0}

\newenvironment{extendedonly}{}{}
\newenvironment{notinextended}{}{}
\newenvironment{masteronly}{}{}
\newenvironment{notinmaster}{}{}

\includecomment{masteronly}
\excludecomment{notinmaster}

\if 1\IsExtended

\includecomment{extendedonly}
\excludecomment{notinextended}

\else

\excludecomment{extendedonly}
\includecomment{notinextended}

\fi

\newtheorem{definition}{Definition}
\newtheorem{theorem}[definition]{Theorem}
\newtheorem{lemma}[definition]{Lemma}
\newtheorem{corollary}[definition]{Corollary}
\newtheorem{example}[definition]{Example}

\theoremstyle{remark}

\newcommand{\shortrules}[6]{\noindent\begin{minipage}{#6ex}{\bfseries #1}\end{minipage} $\;$ #2 $\;\Rightarrow_{\text{#5}}\;$ #3 \par\smallskip\noindent #4}

\newcommand{\comp}{\operatorname{comp}} %formula
\newcommand{\mGnd}{\operatorname{gnd}} %formula
\newcommand{\mMGU}{\operatorname{mgu}} %formula
\newcommand{\dom}{\operatorname{dom}}
\newcommand{\cdom}{\operatorname{codom}}

\newcommand{\mSL}{SCL}
\newcommand{\SL}{{SCL}}
\theoremstyle{definition}

\newcommand{\st}{s.t.\ }

\newcounter{sidenote}

\newcommand{\MEMO}[1]{}

\newcommand{\mGndB}{\ensuremath{\mGnd^{\prec_B \beta}}}

\newcommand{\mrulename}[1]{\ensuremath{}\Rightarrow_{\text{SCL}}^{\text{#1}}}
\newcommand{\rulename}[1]{$\Rightarrow_{\text{SCL}}^{\text{#1}}$}

\hypersetup{breaklinks=true}
\title{SCL(FOL) Revisited}

\author{Martin Bromberger, Simon Schwarz, Christoph Weidenbach \\ \large{Max Planck Institute for Informatics} \\ \large{Saarland Informatics Campus} \\ \large{Saarbr\"ucken, Germany} \\ \large{\texttt{\{mbromber,sschwarz,weidenb\}@mpi-inf.mpg.de}} }

\begin{document}

\maketitle

\begin{abstract}

 This paper presents an up-to-date and refined version of the SCL calculus for first-order logic without equality.
 The refinement mainly consists of the following two parts:
 First, we incorporate a stronger notion of regularity into SCL(FOL). Our regularity definition is adapted from the SCL(T) calculus.
 This adapted definition guarantees non-redundant clause learning during a run of SCL. However, in contrast to the original presentation,
 it does not require exhaustive propagation.  
 Second, we introduce trail and model bounding to achieve termination guarantees. In previous versions, no termination guarantees about SCL were achieved. 
 Last, we give rigorous proofs for soundness, completeness and clause learning guarantees of SCL(FOL) and put SCL(FOL) into context of existing first-order calculi.
  
\end{abstract}

\section{Introduction} \label{sec:introduction}

First-order automated theorem proving is currently a core topic in automated reasoning.
Over the recent years, the development of first-order theorem provers has been very fruitful and had great impact in computer science and mathematics.
For instance, automation of first-order theorem proving provided powerful, fully automated tactics for interactive theorem provers. For example, the ``Sledgehammer''
automation tactic \cite{PaulsonB10,DesharnaisVBW22} provides a link between the interactive prover Isabelle and multiple first-order provers. With this automation tactic,
a significant percentage of nontrivial goals can be proven without interaction. Further research \cite{BlanchetteBN11} employs model-based automated reasoning to not only search for proofs,
but also extract counterexamples in case of falsifiable goals.  %{\color{red}CW: 5x schneller -- Quelle?}
Moreover, automated first-order theorem proving has many use cases in software verification. For instance, progress has been made in the verification of
so-called ``supervisor code'' \cite{BrombergerEtAl21DatalogHammer,BrombergerEtAl22SortedDatalogHammer} by translating properties directly to fragments of first-order logic.
The resulting first-order formula is then automatically checked for validity by a first-order theorem prover.

First-order logic is an important subject of research in automated theorem proving due to multiple beneficial properties of the logic:
First, the logic is very expressive. In particular, computable problems can always be translated to first-order logic.
Second, despite its expressiveness, reasoning in first-order logic can be automated. Last, first-order logic is fairly well-understood and intuitive.
Hence, many problems can be naturally encoded in first-order logic. Note that, in this work, we focus on first-order logic without equality.
SCL for first-order logic with equality is presented in~\cite{LeidingerWeidenbach22IJCAR}.

Overall, the wide applications of automated theorem proving and the advantages of first-order logic explain the demand for efficient first-order automated theorem provers.
This led to a multitude of different methods for first-order reasoning:   %which are introduced in the next section.

\emph{Resolution} and \emph{Superposition} are the two traditional methods for first-order reasoning and both are semi-decision procedures for this logic.
However, they usually employ unguided learning. Hence, sophisticated heuristics need to be used in addition. Furthermore, resolution and
superposition can learn redundant clauses, e.g., clauses that are subsumed by existing clauses. Thus, explicit and costly checks for non-redundancy need to be implemented.

More recent approaches, such as the model-evolution calculus~\cite{BaumgartnerT08DPLLFirstOrder,BaumgartnerFuchsTinelli06} and model-driven conflict searches
in general~\cite{BonacinaEtAl15}, build a candidate model.
Conflicts to this candidate model then guide resolution steps, until either a (partial) model or a refutation is found.
This approach is motivated by the success of propositional model-building approaches in practice, such as DPLL~\cite{DavisLL62DPLL}
and CDCL (conflict-driven clause learning)~\cite{SilvaS96,BayardoSchrag97,MoskewiczMadiganZhaoZhangMalik01}.
In particular, the practical success of propositional CDCL motivates lifting this approach to first-order logic.

The SCL calculus (``Clause Learning from Simple Models'' or short ``Simple Clause Learning'') \cite{FioriWeidenbach19,BrombergerFW21,BrombergerEtAl21DatalogHammer,BrombergerSW22,LeidingerWeidenbach22IJCAR}
lifts a conflict-driven clause learning approach to first-order logic. In its first, original version~\cite{FioriWeidenbach19},
the focus of the calculus is deciding the Bernays-Schoenfinkel class without equality.
By extending this approach, however, SCL even yields a decision procedure for any class that enjoys the finite model property.
Moreover, SCL provides a sound and refutationally complete semi-decision procedure for general first-order logic without equality.
Subsequently, SCL has been extended to handle theories \cite{BrombergerEtAl2020arxiv,BrombergerFW21} and first-order logic with
equality~\cite{LeidingerWeidenbach22IJCAR}. This paper will revisit classical SCL without equality or theories.

SCL builds a partial ground candidate model on the trail with the \emph{model-building} rules. Those rules make use of ground literal decisions and ground literal propagations.
If a conflict between the ground trail and an instantiation of a clause is found, the \emph{conflict resolution} process starts. During this process,
resolution and factoring is applied to the conflict. Finally, a backtracking step learns a new non-redundant clause.

Even though all propagations and decisions on the trail are ground, learned clauses are non-ground in general.
Furthermore, it is guaranteed that a newly learned clause is non-redundant. We call a clause redundant with respect to a clause set if it is already implied by smaller clauses,
where the literal ordering induced by the trail literal ordering is used to compare clauses. 
In SCL, all learned clauses are always non-redundant with respect to our input clause set and previously learned clauses. In practice,
this implies that learned clauses do not need to be checked for forward redundancy when adding them to the learned clause set.
This makes clause learning in SCL a powerful primitive.

Overall, SCL interleaves the trail-building and conflict resolution rules until either a refutation or a model is found.
Similarly to resolution, a clause set is refuted if the empty clause can be learned. The refutational completeness of SCL guarantees that the empty clause is eventually learned for unsatisfiable inputs.
For satisfiable clause sets, SCL will build a model on the trail. In general, the bounded model on the
trail will be only \emph{partial}. 
This is in particular the case if the clause set has only infinite models.
Such a bounded partial model is not necessarily extendable to a complete model, a challenge that
motivates yet another extension of SCL, called HSCL~\cite{BrombergerSW22}. However, in this work, we present a version of SCL
that either proves unsatisfiability or constructs a single partial model on the trail.

In the original work on SCL \cite{FioriWeidenbach19}, two measures of regularity were defined, namely ``regular runs'' and ``weakly-regular runs''. In a regular run,
exhaustive propagation was required. However, in combination with the ground trail of SCL, this can already lead to exponentially many propagations.
For example, consider the following clause set:
$$N = \{R(x_1,\ldots,x_n,a,b), P \lor Q, P \lor\lnot Q, \lnot P \lor Q, \lnot P \lor\lnot Q\}$$
In the case of exhaustive unit propagation, this example requires propagating $2^n$ different ground instances of $R(x_1,\ldots,x_n,a,b)$ before being
able to refute the propositional logic clauses. In this example, it is clear that all ground instances of $R(x_1,\ldots,x_n,a,b)$ may be ignored.
In general, however, this may be not the case. Hence, this motivates a calculus without mandatory propagation in order to learn non-redundant clauses.

The original paper~\cite{FioriWeidenbach19} addressed this issue by defining weakly-regular runs. In such runs, unit propagation is treated differently and is not required to be exhaustive. However,
all non-unit propagations still need to be applied before any decision can be made. While this mitigates the exponential growth in the above example,
more complex examples still exhibit unnecessary exponential trail growth.

This motivates the use of a different regularity measure. The measure used in this work was first presented in \cite{BrombergerEtAl2020arxiv}.
As a key idea, regularity does no longer require exhaustive propagation. Instead, decisions are limited in a way that does not allow conflicts to arise directly after a decision is made.
We will show that this regularity, although weaker as the previous definition, is still sufficient to guarantee non-redundant learning (see Lemma~\ref{lemma:resolve_in_regular}).
Our contribution is to adapt the new regularity definition \cite{BrombergerEtAl2020arxiv} to the SCL calculus, which includes refining the proofs to the weaker definition of regularity.
Furthermore, the variant of the Backtracking rule, Section~\ref{sec:rules}, used in this paper is more explicit compared to \cite{BrombergerEtAl2020arxiv} because it jumps to
the minimal trail where the finally learned clause propagates.

In the original SCL paper \cite{FioriWeidenbach19}, no trail bounding was employed at all. Hence, termination is only guaranteed in the case of the underlying logic enjoying the finite model property.
Consequently, in the original version, the analysis focused on the impact of SCL for the Bernays-Schoenfinkel fragment, since problems from this fragment are guaranteed to have a finite model.

A contribution of this paper is to use a well-founded ordering $\prec_B$ and a limiting literal $\beta$ instead of a fixed set $B$. We require only finitely
many literals to be $\prec_B$-smaller than $\beta$ (a formal definition is in Section \ref{sec:prelim}). Now, as an invariant,
we will prove that all considered literals are $\prec_B$-smaller than $\beta$. Hence, only finitely many literals can occur during a SCL run.
Thus, we obtain guarantees about finite trails and model sizes, while avoiding the problems of a fixed set of constants. Overall,
this definition is used to guarantee termination for the bounded variant of SCL. Completeness is achieved via a new Grow rule that strictly increases
$\beta$, see Section~\ref{sec:rules}.

Lastly, this paper contains rigorous proofs that were omitted in the original papers \cite{FioriWeidenbach19,BrombergerEtAl2020arxiv,LeidingerWeidenbach22IJCAR} before,
as well as some simplified proofs. In particular, this paper includes a full soundness proof of SCL (Theorem~\ref{theo:scl:soundness}) and a rigorous correctness proofs (Theorem~\ref{theo:scl:corterm}).
Further proofs have been simplified, for example the complexity analysis in Theorem~\ref{theo:cslnonrednexptime} and the termination proof (Theorem~\ref{theo:finite-termination}).

The paper is now organized as follows.
First, in Section \ref{sec:related}, we compare four different approaches to first-order reasoning and put SCL into context.
To this end, we introduce classical resolution and superposition.
Afterwards, we compare these with more recent approaches, namely the model-evolution search and model-driven conflict searches.
In the preliminaries, Section~\ref{sec:prelim}, we then briefly introduce the formal notation and definitions used in the main part of this paper.
Next, we formally describe the classical SCL calculus for first-order logic without equality in Section \ref{sec:rules}. First, we give the abstract rewrite rules for the calculus.
These rules follow the ones presented in the original SCL paper \cite{FioriWeidenbach19}. However, the rules have been adapted to
(i)~respect a new clause size measure $\prec_B$ to limit the trail size and guarantee termination of the calculus
(ii)~directly incorporate our new definition of regularity. %, instead of defining these properties on runs.
Then we give full proofs for the core properties of SCL with respect to the changed definition of the rules. In particular,
this includes soundness (Theorem \ref{theo:scl:soundness}), termination (Theorem \ref{theo:finite-termination}) and non-redundant learning (Theorem \ref{theo:SnnRed}).
Lastly, we will discuss our contributions and potential future work in Section \ref{sec:conclusion}.

%%% Local Variables:
%%% mode: latex
%%% TeX-master: "paper"
%%% End:

\section{Related Work} \label{sec:related}

There has been intensive development on reasoning procedures for first-order logic without equality. In this section, we will compare four different approaches to first-order reasoning.

Traditionally, the work in first-order-logic automated reasoning has focused on proofs. From this perspective, we first introduce resolution and superposition,
both of which are refutation-based algorithms which can produce (refutation) proofs, but yield no easily accessible model information in general.
Hence, in a prover setting, extracting counterexamples from runs of those algorithms is not straight-forward.

Afterwards, we focus on model-based calculi. These introduce a candidate partial model. This model then guides reasoning, usually in
the form of finding an ordering in which resolution inferences are made. As examples for model-based calculi, we present the model-evolution
calculus and the class of model-driven conflict searches. The SCL calculus, which will be detailed after this section, implements a model-driven conflict search.

Resolution \cite{DavisP60Resolution,Robinson65,Slagle67} is one of the classical approaches to first-order theorem proving.
In the classical resolution calculus, two clauses can produces a new clause (called resolvent) by the resolution rule. For example,
the clause $P(x) \lor Q(x)$ together with the clause $\neg P(g(y))\lor R(y,y)$ produces the resolvent $Q(g(y))\lor R(y,y)$.
Furthermore, resolution requires the concept of \emph{factoring} for completeness. For example, a clause $R(x,y) \lor P(x) \lor P(g(y))$ can be factorized to $R(g(y),y) \lor P(g(y))$
by unification of $P(x)$ and $P(g(y))$. Exhaustive application of resolution and factoring modulo certain redundancy criteria is called \emph{saturation}.
The most important redundancy criteria in the resolution context are tautology deletion, i.e., deletion of clauses like $R(x,y)\lor P(x)\lor \neg P(x)$, and the
deletion of subsumed clauses. For example, the clause $P(x)\lor R(x,y)$ \emph{subsumes} the clause $P(x)\lor R(x,g(z)) \lor P(z) $ via the substitution $\{y\mapsto g(z)\}$
because $(P(x)\lor R(x,y))\{y\mapsto g(z)\}$ is a subset of $P(x)\lor R(x,g(z)) \lor P(z) $.
A clause set without equality can already be saturated by applying only the resolution and factoring rule. In a saturated clause set, no further non-redundant clause can be added.
 Resolution has \emph{refuted} the clause set if the empty clause can be derived. In particular, a clause set is unsatisfiable iff the empty clause is in the saturated clause set. 

\medskip\noindent
As an example, consider the following (unsatisfiable) clause set:
%\TODO{erste resolution mit clause nicht im Set}
%\TODO{konstante + funktionsterm in dieses ding rein}
\begin{align*}
%N =
\{P(a) \lor Q(x), \hspace*{1em}P(f(x)) \lor \neg Q(x), \hspace*{1em}\neg P(x) \lor  Q(x), \hspace*{1em}\neg P(f(a)) \lor \neg Q(x)\}
\end{align*}
Here, a possible resolution refutation could look as follows:
\begin{align*}
    P(a) \lor Q(x) \text{~resolved with~} \neg P(x) \lor Q(x) &\leadsto Q(x) \lor Q(a) \\
    P(f(x)) \lor \neg Q(x) \text{~resolved with~} \neg P(f(a)) \lor \neg Q(x) &\leadsto \neg Q(a) \lor \neg Q(x) \\
    \text{Factoring on~}Q(x) \lor Q(a)\{x \mapsto a\} &\leadsto Q(a) \\
    \text{Factoring on~}\neg Q(a) \lor \neg Q(x)\{x \mapsto a\} &\leadsto \neg Q(a) \\
    Q(a) \text{~resolved with~} \neg Q(a) &\leadsto \bot
\end{align*}

Resolution is a refutation-based approach. Hence, to prove validity of a formula, a contradiction from the negation
of the statement must be derived. By the compactness of first-order logic, for an unsatisfiable clause set such a contradiction
is eventually found. This makes resolution a semi-decision procedure for first-order logic without equality.
However, in satisfiable cases, no direct approach for extracting a model from the resolution calculus exists. This is a drawback
compared to explicit model-building approaches, where (partial) models can be extracted easily. %Hence, in a prover setting, a counter-example cannot be extracted from resolution in a straight-forward way.

Furthermore, resolution employs \emph{unguided} learning. Multiple different inferences can be done at any step in the algorithm, and there are no limits with respect to generating new clauses.
While this enables resolution to potentially find optimal inferences which yield very short resolution proofs, there are several downsides to this approach:

First, not limiting inferences leaves the burden of choice to the prover. Hence, resolution-based solvers need a sophisticated
heuristic for priorizing inferences. This is one of the reasons that led to the development of superposition, which is discussed in the next section.

Second and more importantly, redundant clauses may be generated by resolution inferences. For example, consider the same clause set as in the initial example.
%\centerline{$N = \{P(x) \lor Q(x), \hspace*{1em}P(x) \lor \neg Q(x), \hspace*{1em}\neg P(x) \lor \neg Q(x), \hspace*{1em}\neg P(x) \lor \neg Q(x)\}$}
While we started with a reasonable inference in the initial example, a potential resolution run could also take the following inference:
\begin{align*}
    P(x) \lor Q(x) \text{~resolved with~} \neg P(x) \lor \neg Q(x) &\leadsto Q(x) \lor \neg Q(x)
\end{align*}

However, the generated clause does not contribute to any refutation, as it is a tautology. Later in this work, we will
formalize a similar concept called \emph{clause redundancy}. In the above example, resolution generated a redundant clause. In general, this can happen also for inferred clauses that are not tautological.
Hence, resolution-based provers must check for redundancy whenever a new clause is generated. This can have a significant
impact on the runtime of the prover, even though there is active development on efficient subsumption testing \cite{BrombergerLW22}.
In contrast, SCL avoids this checking entirely by providing native non-redundant learning guarantees. %Hence, all learned clauses cannot be redundant with respect to ou.
%Still, the classical resolution calculus forms the basis of superposition and is as such implemente
Resolution can be extended to handle equality by adding the equality axioms \cite{BachmairG98equationalreasoning}.
Further improvements, subsequently, have been made to exploit the equality axioms for generating smaller clauses. One such example is
%\TODO{superposition, ganzinger papier -- dann ist das eine ueberleitung}
superposition \cite{BachmairG90Superpos}, which refines the idea of paramodulation \cite{Robinson1983}. With this approach, equalites are oriented and equal atoms in terms can be replaced.

Superposition approaches \cite{BachmairG90Superpos,BachmairGW94,KruglovW12,BaumgartnerW19} limit the possible inferences of resolution by employing a fixed \emph{literal ordering}.
They were defined originally for first-order logic with equality and their projection to first-order logic without equality considered here is called \emph{ordered resolution}.
With respect to the literal  ordering, only resolution inferences on maximal literals are permitted.
This severely limits the amount of possible inferences, which can significantly speed up saturation. In contrast to resolution there is an abstract concept of redundancy:
a clause is redundant and can therefore be eliminated, if it is implied by smaller clauses with respect to the literal ordering.
The restriction of resolving only on maximal literals preserves refutational completeness.
However, even though there are fewer inferences possible, cases in which multiple inferences are available still occur frequently.
Hence, a heuristic for priorization is still necessary. Overall, the order in which inferences are made has a significant impact on the runtime.
Moreover, the performance of superposition heavily depends on the used term ordering~\cite{Dershowitz82,Schulz22TermOrderings}, as it dictates the order of possible inferences.
This poses a challenge for an implementation, as such orderings usually must be fixed before starting a superposition run. Furthermore, the problem of inferring redundant clauses persists in superposition.
Nonetheless, superposition-based solvers are state-of-the-art for first-order reasoning. During recent CASCs (CADE ATP System Competitions) \cite{Sutcliffe16} superposition-based provers
have dominated the first-order tracks \cite{DuarteK20iProver,Kovacs13Vampire,Schulz19EProver}. % zipperposition?
Similarly to resolution, superposition can be extended to handle first-order logic with equality. Most state-of-the-art provers include equational reasoning to support first-order logic with equality.

%\subsection{Hypertableaux}f
% \cite{Baumgartner98}

An alternative to resolution-based reasoning is formed by model-based approaches. This is motivated by the good practical performance of
DPLL~\cite{DavisLL62DPLL} in propositional logic.
In DPLL, the resolution technique is replaced by splitting on literals. On the respective literal, a case analysis is then performed:
the possibility of it being false or true are investigated separately. This can also be understood as ``guessing'' or deciding the specific literal.
The model-evolution calculus \cite{BaumgartnerT08DPLLFirstOrder} lifts this approach to first-order logic.
%The key challenge while lifting DPLL to first-order is to appropriately handle clause splitting.
% This introduces new Skolem constants, which do not unify with other variables.
In the model-evolution calculus, contexts are introduced. A context always forms a candidate model for the input clauses. On conflicts, the model is ``evolved''
by splitting on conflicting literals. One key challenge is an efficient handling of clause splitting. In particular, in contrast to propositional logic,
a single first-order clause represents all its ground instances. Hence, when splitting $P(x)$, both a branch $P(x)$ (for $\forall x P(x)$)
and a branch $\neg P(c)$ (for $\neg \forall x P(x)$) need to be produced. In particular, note that a new (skolemized) variable is introduced in the existential case.
The model-evolution calculus introduces techniques to handle splitting efficiently, by storing additional information about existential variables in a context.

Another instance of a model-building reasoning approach is NRCL (Non-Redundant Clause Learning)~\cite{AlagiWeidenbach15}. Actually, this is the predecessor of SCL
with the main difference that model building in NRCL is more complex, because also non-ground literals are propagated. This then results in more complex algorithms
for propagation or false clause detection, similar to~\cite{PiskacEtAl10}.

Most propositional solvers currently use variants of propositional CDCL \cite{SilvaS96,BayardoSchrag97} instead of DPLL.
Motivated by the success of CDCL solvers in practice \cite{Biere08CDCLSolvers}, there are approaches of lifting such algorithms to the first-order case.
%Mainly, those implementations aim to preserve the following key features of CDCL:
Such a lifting preserves the general ideas of CDCL: First, a \emph{trail} of literals is built.
The trail forms a candidate model and must always be consistent in itself.
The algorithm then assumes that all literals on the trail are true, whereas literals not on the trail are undefined.
Then, the trail is extended until either a model is found, or a \emph{conflict} arises. This happens when a clause becomes false under the literals in the trail.
Such a conflict is then \emph{resolved} until a new clause is learned.
Compared to superposition, this approach allows choosing resolution inferences without an a priori ordering. Instead, inferences are guided by the model assumptions on the trail.
In this approach, model extraction is straight-forward, as a (partial) model is always present on the trail in case of a non-active conflict.
This technique is employed by the SCL calculus.
Furthermore, there is active development on the first-order SGGS calculus (Semantically-Guided Goal-Sensitive Reasoning) \cite{BonacinaPlaisted16}, which also implements a model-driven search.
However, the strong non-redundant clause learning guarantees of SCL are not provided by SGGS.
Another key difference in the design of both calculi is that SCL employs a ground trail %and has non-redundant clause learning guarantees,
%\TODO{non redundant learning ist auch eine key difference}
whereas SGGS can add non-ground literals to the trail. While this makes the trail representation of SGGS more compact and powerful,
this comes at the cost of more expensive unification steps for propagation, conflict detection and resolution. %Moreover, in contra

\section{Preliminaries} \label{sec:prelim}

The general notation and definitions follow the original SCL paper \cite{FioriWeidenbach19}.
We assume a first-order language without equality where
$N$ denotes a clause set;
$C, D$ denote clauses;
$L, K, H$ denote literals;
$A, B$ denote atoms;
$P, Q, R$ denote predicates;
$t, s$ terms;
$f, g, h$ function symbols;
$a, b, c$ constants;
and $x, y, z$ variables.
Atoms, literals, clauses and clause sets are considered as usual, where
in particular clauses are identified both with their disjunction and multiset
of literals.
The complement of a literal is denoted by the function $\comp$.
Semantic entailment $\models$ is defined as usual where variables in clauses
are assumed to be universally quantified.
Substitutions $\sigma, \tau$ are total mappings from variables to terms, where
$\dom(\sigma) := \{x \mid x\sigma\neq x\}$ is finite and $\cdom(\sigma) := \{ t\mid x\sigma = t, x\in\dom(\sigma)\}$.
Their application is extended to literals, clauses, and sets of such objects in the usual way.
A term, atom, clause, or a set of these objects is \emph{ground} if it does not contain any variable.
A substitution $\sigma$ is \emph{ground} if $\cdom(\sigma)$ is ground. A substitution $\sigma$ is \emph{grounding}
for a term $t$, literal $L$, clause $C$ if $t\sigma$, $L\sigma$, $C\sigma$ is ground, respectively.
The function $\mMGU$ denotes the \emph{most general unifier} of two terms, atoms, literals.
We assume that any $\mMGU$ of two terms or literals does not introduce any fresh variables and is idempotent.
A \emph{closure} is denoted as $C\cdot\sigma$ and is a pair of a clause $C$ and a grounding substitution $\sigma$.
The function $\mGnd$ returns the set of all ground instances of a literal, clause, or clause set with respect
to the signature of the respective clause set.

%Let \emph{BS} denote the Bernays-Sch\"onfinkel clause set fragment, where the only terms in
%a clause set are variables and constants.
%Note that for BS the set of ground terms that needs to be considered is always finite, whereas for first-order logic it is infinite, in general.

Let  $\prec$ denote a well-founded, total, strict ordering on
ground literals.
This ordering is then lifted to clauses and clause sets by its respective multiset extension. We
overload $\prec$ for literals, clauses, clause sets if the meaning is clear from the context.
The ordering is lifted to the non-ground case via instantiation: we define $C \prec D$
if for all grounding substitutions $\sigma$ it holds $C\sigma \prec D\sigma$.
We define $\preceq$ as the reflexive closure of $\prec$ and $N^{\preceq C} := \{D \mid D\in N \;\text{and}\; D\preceq C\}$.

\begin{definition}[Clause Redundancy] \label{prelim:def:redundancy}
  A ground clause $C$ is \emph{redundant} with respect to a ground clause
  set $N$ and an order $\prec$ if $N^{\prec C} \models C$ or $C\in N$.
  A clause $C$ is \emph{redundant}  with respect to a clause set
  $N$ and an order $\prec$  if for all $C' \in \mGnd(C)$ it holds that $C'$ is
  redundant with respect to $\mGnd(N)$.
\end{definition}

%The following definitions for bounding have been presented at the PAAR workshop .
For the sake of bounding, let $\prec_B$ denote a  well-founded, total, strict ordering on
ground atoms such that for any ground atom $A$ there are only
finitely many ground atoms $A'$ with $A' \prec_B A$ \cite{BrombergerSW22}.
For example, an instance of such an ordering could be KBO~\cite{KnuthBendix70} without zero-weight symbols.
In contrast, a LPO~\cite{Dershowitz82} does not fulfil this property if there are function symbols of non-zero arity.
The ordering $\prec_B$ is lifted to literals
by comparing the respective atoms. It is lifted to clauses by a multiset extension.
Given an ordering $\prec_B$ and a ground literal $\beta$, the function $\mGndB$ computes the set of all ground instances
of a literal, clause, or clause set where the grounding is restricted to produce literals $L$ with $L \prec_B \beta$.

A ground clause $C$ is true in a model
$M$, denoted $M\models C$, if $C\cap M \not= \emptyset$. Conversely, a clause $C$ is false in $M$ if $\{\comp(L) \mid L \in C\} \subseteq M$. Otherwise, the clause is undefined in $M$. A ground clause set $N$ is true in $M$, denoted $M\models N$ if
all clauses from $N$ are true in $M$.

%%% Local Variables:
%%% mode: latex
%%% TeX-master: "paper"
%%% End:

\renewcommand{\mrulename}[1]{\ensuremath{}\Rightarrow_{\text{SCL}}^{\text{#1}}}
\renewcommand{\rulename}[1]{$\Rightarrow_{\text{SCL}}^{\text{#1}}$}

\section{SCL Rules and Properties} \label{sec:rules}

The presentation of the SCL rules and properties is following the original SCL presentation~\cite{FioriWeidenbach19}.
However, the SCL rules have been modified to incorporate a different regularity measure, which was first presented in \cite{BrombergerEtAl2020arxiv}.
Furthermore, the rules of the given SCL calculus are changed to support a bounding of the trail size. To this extent, we use a limiting literal $\beta$ and a
well-founded ordering $\prec_B$ which restricts all considered literals.

After the presentation of the SCL rules, we prove the key properties of SCL. These properties include soundness (Theorem~\ref{theo:scl:soundness}),
refutational completeness (Theorem~\ref{theo:scl:refcomplete}), termination (Theorem~\ref{theo:finite-termination}) and non-redundant clause learning (Theorem~\ref{theo:SnnRed}).
In particular, contributions of this paper are a full soundness proof, a rigorous correctness proof (Theorem~\ref{theo:scl:corterm}),
as well as simplifications of the other proofs. Furthermore, all properties and proofs have been refined to adhere to the changed regularity definition as well as trail bounding.

The inference rules of SCL are represented by
an abstract rewrite system.
They operate on a problem state, a six-tuple
$(\Gamma; N; U; \beta; k; D)$ where $\Gamma$ is a sequence
of annotated ground literals, the \emph{trail};
$N$ and $U$ are the sets of \emph{initial} and \emph{learned}
clauses; $\beta$ is a ground literal limiting the size of the trail; $k$ counts the number of decisions; and
$D$ is a status closure that is either true $\top$, false $\bot$,
or $C\cdot\sigma$.
Literals in $\Gamma$ are either annotated with
a number, also called a level; i.e., they have the form $L^k$
meaning that $L$ is the $k$-th guessed decision
literal, or they are annotated with a closure that
propagated the literal to become true.
A ground literal $L$ is of
\emph{level} $i$ with respect to a problem state
$(\Gamma; N; U; \beta; k; D)$ if $L$ or $\comp(L)$ occurs
in $\Gamma$ and the first decision literal left from
$L$ ($\comp(L)$) in $\Gamma$, including $L$, is annotated with $i$.
If there is no such decision literal then its level
is zero. A ground clause $D$ is of \emph{level} $i$
with respect to a problem state $(\Gamma; N; U; \beta; k; D)$
if $i$ is the maximal level of a literal in $D$. The level of the empty clause $\bot$ is 0.
Recall $D$ is a non-empty closure or $\top$ or $\bot$.
Similarly, a trail $\Gamma$ is of level $i$ if the maximal literal in $\Gamma$ is of level $i$.

A literal $L$ is \emph{undefined} in $\Gamma$
if neither $L$ nor $\comp(L)$ occur in $\Gamma$.
We omit annotations to trail literals if they play no role in the respective context.
Initially, the state
for a first-order clause set $N$ is $(\epsilon; N; \emptyset; \beta; 0; \top)$.

\medskip\noindent
The rules for conflict search are:

\bigskip
\shortrules{Propagate}
{$(\Gamma;N;U;\beta;k;\top)$}
{$(\Gamma, L\sigma^{(C_0\lor L){\delta}\cdot\sigma};N;U;\beta;k;\top)$}
{provided $C\lor L\in (N\cup U)$, $C = C_0 \lor C_1$, $C_1\sigma = L\sigma \lor \dots \lor L\sigma$,
  $C_0\sigma$ does not contain $L\sigma$, {$\delta$ is the mgu of the literals in $C_1$ and $L$}, $(C\lor L)\sigma$ is ground, $(C\lor L)\sigma \prec_B \{\beta\}$,
  $C_0\sigma$ is false under $\Gamma$, and $L\sigma$ is undefined in $\Gamma$}{\mSL}{12}

\bigskip
The rule Propagate applies exhaustive factoring to the propagated literal with respect to the grounding substitution $\sigma$ and annotates the factored clause to the propagation
literal on the trail.

\bigskip
\shortrules{Decide}
{$(\Gamma;N;U;\beta;k;\top)$}
{$(\Gamma,L\sigma^{k+1};N;U;\beta;k+1;\top)$}
{provided $L\in C$ for a $C \in (N\cup U)$, $L\sigma$ is a ground literal undefined in $\Gamma$, and  $L\sigma \prec_B \beta$}{\mSL}{12}

\bigskip
\shortrules{Conflict}
{$(\Gamma;N;U;\beta;k;\top)$}
{$(\Gamma;N;U;\beta;k;D\cdot\sigma)$}
{provided $D\in (N\cup U)$, $D\sigma$ false in $\Gamma$
  for a grounding substitution $\sigma$}{\mSL}{12}

\bigskip
These rules construct a (partial) model via the trail $\Gamma$ for $N\cup U$ until a conflict, i.e.,
a false clause with respect to $\Gamma$ is found. The above rules always terminate, because there are
only finitely many ground literals $K$ with $K\prec_B \beta$. Choosing an appropriate $\beta$ is
sufficient for completeness for certain first-order fragments, e.g. the BS fragment. In particular, for any
fragment with the finite model property, completeness can be achieved with SCL for appropriate $\beta$.
In general, a rule Grow~\cite{FioriWeidenbach19}, see below, increasing $\beta$ is needed for full first-order completeness.
In the special case of a unit clause $L$, the rule Propagate actually
annotates the literal $L$ with a closure of itself. So the propagated literals on the trail
are annotated with the respective propagating clause and the decision literals with
the respective level.
If a conflict is found, it is resolved by the rules below.
Before any Resolve step, we
assume that the respective clauses are renamed such that they do not share any variables and
that the grounding substitutions of closures are adjusted accordingly.

% \bigskip
% \shortrules{Skip-NoCM}
% {$(\Gamma, L\delta^{(C\lor L) \cdot\delta};N;U;\beta;k;D\cdot\sigma)$}
% {$(\Gamma;N;U;\beta;k;D\cdot\sigma)$}
% {provided $\comp(L\delta)$ does not occur in $D\sigma$}{\mSL}{11}

\bigskip
\shortrules{Skip}
{$(\Gamma, L;N;U;\beta;k;D\cdot\sigma)$}
{$(\Gamma;N;U;\beta;k-i;D\cdot\sigma)$}
{provided $\comp(L)$ does not occur in $D\sigma$, if $L$ is a decision literal then $i=1$, otherwise $i=0$}{\mSL}{11}

% NOT needed with CM version of rules
% \bigskip
% \shortrules{Clean}
% {$(\Gamma, K^k, \Gamma';N;U;\beta;k;D\cdot\sigma)$}
% {$(\Gamma, \Gamma'';N;U;\beta;k-1;D\cdot\sigma)$}
% {provided $D\sigma$ is false in $\Gamma, \Gamma''$ and $\Gamma' = [K_1,\ldots,K_n]$, the subsequence $\Gamma''$ of $\Gamma'$ is such that
%  all literals $K_i$ in $\Gamma'' = \Gamma''_1, K_i, \Gamma''_2$ propagate with respect to $\Gamma \Gamma''_1$}{\mSL}{11}

\bigskip
\shortrules{Factorize}
{$(\Gamma;N;U;\beta;k;(D\lor L \lor L')\cdot\sigma)$}
{$(\Gamma;N;U;\beta;k; (D\lor L)\eta\cdot\sigma)$}
{provided $L\sigma = L'\sigma$, $\eta=\mMGU(L,L')$}{\mSL}{11}

% \bigskip
% \shortrules{Resolve-NoCM}
% {$(\Gamma, L\delta^{(C\lor L)\cdot\delta};N;U;\beta;k;(D\lor L')\cdot\sigma)$}
% {$(\Gamma, L\delta^{(C\lor L)\cdot\delta};N;U;\beta;k;(D\lor C)\eta\cdot\sigma\delta)$}
% {provided $D\sigma$ is of level $k$, $L\delta = \comp(L'\sigma)$,
%   $\eta=\mMGU(L,\comp(L'))$}{\mSL}{11}

\bigskip
\shortrules{Resolve}
{$(\Gamma, L\delta^{(C\lor L)\cdot\delta};N;U;\beta;k;(D\lor L')\cdot\sigma)$ \\ \hspace*{2.46em} }
{$(\Gamma, L\delta^{(C\lor L)\cdot\delta};N;U;\beta;k;(D\lor C)\eta\cdot\sigma\delta)$}
{provided $L\delta = \comp(L'\sigma)$,
 $\eta=\mMGU(L,\comp(L'))$}{\mSL}{13}

% \bigskip
% \shortrules{Backtrack-NoCM}
% {$(\Gamma,K^{i+1},\Gamma';N;U;\beta;k;(D\lor L)\cdot\sigma)$}
% {$(\Gamma;N;U\cup\{D\lor L\};i;\top)$}
% {provided $L\sigma$ is of level $k$ and $D\sigma$ is of level $i$.}{\mSL}{11}

\bigskip
\shortrules{Backtrack}
{$(\Gamma_0,K,\Gamma_1,\comp(L\sigma)^k;N;U;\beta;k;(D\lor L)\cdot\sigma)$ \\ \hspace*{2.8em}}
{$(\Gamma_0;N;U\cup\{D\lor L\};\beta;j;\top)$}
{provided $D\sigma$ is of level $i'<k$,
 and $\Gamma_0,K$ is the minimal trail subsequence such that there is
 a grounding substitution $\tau$ with $(D \lor L)\tau$ is false in $\Gamma_0,K$ but not in  $\Gamma_0$, and $\Gamma_0$ is of level $j$}{\mSL}{13}

 % \TODO{überspringen wir hier zu viel?}

\bigskip
Please note that it is not always sufficient to backtrack to the smallest decision level from which $(D\lor L)\cdot\sigma$ can propagate (i.e., the case where $\tau = \sigma$ and $i' = j$).
The reason is that there might be other groundings $\tau$ for which $(D\lor L) \cdot \tau$ is now conflicting at level $i'$ due to Resolve and Factorize steps.
Therefore, SCL backtracks to the smallest decision level from which $(D\lor L)$ can propagate for any grounding $\tau$.
The clause $D\lor L$
added by the rule Backtrack to $U$ is called a
\emph{learned clause}.
The empty clause $\bot$ can only be generated by
rule Resolve or be already present in $N$, hence, as usual for CDCL style calculi, the generation
of $\bot$ together with the clauses in $N\cup U$ represent a resolution refutation.

\begin{example}

For example, consider the following clause set:
$$N = \left\{ \begin{array}{ll}
  C_1 = P(x) \lor Q(b), \hspace{0.6em} &C_2 = P(x) \lor \neg Q(y), \\
  C_3 = \neg P(a) \lor Q(x), \hspace{0.6em} &C_4 = \neg P(x) \lor \neg Q(b)

\end{array}
\right\}$$

To refute this clause set, we choose $\beta = R(b)$ and $\prec_B$ is a KBO with all symbols having unit weight and with precedence $a \prec b \prec P \prec Q \prec R$. Hence, it holds that $\{ P(a), P(b), Q(a), Q(b)\} \prec_B \{\beta\}$.
A possible SCL refutation for $N$ could look as follows:

\renewcommand{\arraystretch}{1.2}
\[
 \begin{array}[]{ll}
 & (\varepsilon; N; \emptyset; \beta; 0; \top) \\
 \mrulename{Decide}     & (\neg P(a)^1; N; \emptyset; \beta; 1; \top) \\
 \mrulename{Propagate}  & (\neg P(a)^1 \neg Q(b)^{C_2 \cdot  \{x \mapsto a, y \mapsto b\}}; N; \emptyset; \beta; 1; \top) \\
 \mrulename{Conflict}  & (\neg P(a)^1 \neg Q(b)^{C_2 \cdot \{x \mapsto a, y \mapsto b\}}; N; \emptyset; \beta; 1; P(x) \lor Q(b) \cdot \{ x \mapsto a \}) \\
 \mrulename{Resolve}  & (\neg P(a)^1 \neg Q(b)^{C_2 \cdot \{x \mapsto a, y \mapsto b\}}; N; \emptyset; \beta; 1; P(x) \lor P(x) \cdot \{ x \mapsto a, y \mapsto b \}) \\
 \mrulename{Factorize}  & (\neg P(a)^1 \neg Q(b)^{C_2 \cdot \{x \mapsto a, y \mapsto b\}}; N; \emptyset; \beta; 1; P(x) \cdot \{ x \mapsto a, y \mapsto b \}) \\
 \mrulename{Skip}  & (\neg P(a)^1; N; \emptyset; \beta; 1; P(x) \cdot \{ x \mapsto a, y \mapsto b \}) \\
 \mrulename{Backtrack}  & (\varepsilon; N; \{P(x)\}; \beta; 0; \top) \\
 \mrulename{Propagate}  & (P(a)^{P(x) \cdot \{x \mapsto a\}}; N; \{P(x)\}; \beta; 0; \top) \\
 \mrulename{Propagate}  & (P(a)^{P(x) \cdot \{x \mapsto a\}} Q(b)^{C_3 \cdot \{x \mapsto b\}}; N; \{P(x)\}; \beta; 0; \top) \\
 \mrulename{Conflict}  & (P(a)^{P(x) \cdot \{x \mapsto a\}} Q(b)^{C_3 \cdot \{x \mapsto b\}}; N; \{P(x)\}; \beta; 0; \neg P(x) \lor \neg Q(b) \cdot \{x \mapsto a\}) \\
 \mrulename{Resolve}  & (P(a)^{P(x) \cdot \{x \mapsto a\}} Q(b)^{C_3 \cdot \{x \mapsto b\}}; N; \{P(x)\}; \beta; 0; \neg P(x) \lor \neg P(a) \cdot \{x \mapsto a\}) \\
 \mrulename{Skip}  & (P(a)^{P(x) \cdot \{x \mapsto a\}}; N; \{P(x)\}; \beta; 0; \neg P(x) \lor \neg P(a) \cdot \{x \mapsto a\}) \\
 \mrulename{Factorize}  & (P(a)^{P(x) \cdot \{x \mapsto a\}}; N; \{P(x)\}; \beta; 0; \neg P(a) \cdot \{x \mapsto a\}) \\
 \mrulename{Resolve}  & (P(a)^{P(x) \cdot \{x \mapsto a\}}; N; \{P(x)\}; \beta; 0; \bot \cdot \{x \mapsto a\}) \\
\end{array}
\]

Note the resolution steps are guided by the trail, but always happen between the non-ground original clauses. This allows SCL to learn the non-ground unit clause $P(x)$ in this derivation.
\end{example}

The rules for SCL are applied in a don't-care style, hence, the calculus
offers freedom with respect to factorization. Literals in the conflict clause
can, but do not have to be factorized. In particular, the Factorize rule
may remove duplicate literals. This freedom can result in different learned clauses, see Example \ref{ex:factorize_different}.
The rule Resolve does not remove the literal resolved upon from the trail.
Actually, Resolve is applied as long as the rightmost propagated trail literal
occurs in the conflict clause. This literal is eventually removed by rule Skip from the
trail.

\begin{example} \label{ex:factorize_different}
For example, consider the clause set
\begin{align*}
  N = \left\{
    \begin{array}[]{l}
      D = Q \lor R(a,y)\lor R(x,b) \\
      C = Q\lor S(x,y)\lor P(x)\lor P(y)\lor \lnot R(x,y)
    \end{array}
    \right\}
\end{align*}
and a problem state:
$$([\lnot P(a)^1, \lnot P(b)^2,\lnot S(a,b)^3,\lnot Q^4, \lnot R(a,b)^{C\cdot\{x\mapsto a, y\mapsto b\}}];N;\emptyset;\neg R(b,b); 4;\top)$$
derived by SCL. We assume $\neg R(b,b)$ to the largest literal among all ground instances of $P$, $S$, $Q$, $R$ literals over the constants $a$, $b$.
The rule Conflict is applicable and yields the conflict state
\begin{align*}
  (\lnot P(a)^1, \lnot P(b)^2,\lnot S(a,b)^3,\lnot Q^4, \neg R(a,b)^{C\cdot\{x\mapsto a, y\mapsto b\}}
  ; \\ N;\emptyset;\neg R(b,b); 4;
  D\cdot\{x\mapsto a, y\mapsto b\})
\end{align*}
from which we can either learn the clause
  $$C_1=Q\lor S(x,b)\lor P(x)\lor P(b)\lor S(a,y)\lor P(a)\lor P(y)$$
  or the clause
  $$C_2=Q\lor S(a,b)\lor P(a)\lor P(b)$$
  depending on whether we first resolve or factorize. Note that $C_2$ does not subsume $C_1$. Both clauses are non-redundant.
  In order to learn $C_1$ we need to resolve twice with $R(a,b)^{C\cdot\{x\mapsto a, y\mapsto b\}}$.
\end{example}

The first property we prove about SCL is soundness. We prove it via the notion of a sound state.

\begin{definition}[Sound States]
  \label{def:sound_states}

  A state $(\Gamma;N;U;\beta;k;D)$ is \emph{sound} if the following conditions hold:
  \begin{enumerate}
    \item[1.] $\Gamma$ is a consistent sequence of annotated ground literals, i.e. for a ground literal $L$ it cannot be that $L \in \Gamma$ and $\neg L \in \Gamma$
    \item[2.] for each decomposition
      $\Gamma = \Gamma_1,L\sigma^{C\lor L\cdot\sigma},\Gamma_2$ we have that $C\sigma$ is false under $\Gamma_1$ and
      $L\sigma$ is undefined under $\Gamma_1$, and $N\cup U \models C\lor L$,
    \item[3.] for each decomposition $\Gamma = \Gamma_1,L^k,\Gamma_2$
      we have that $L$ is undefined in $\Gamma_1$,
    \item[4.] $N\models U$,
    \item[5.] if $D=C\cdot\sigma$ then $C\sigma$ is false
    under $\Gamma$ and $N\models C$. In particular, $\mGndB(N) \models C\sigma$,
    \item[6.] for any $L\in\Gamma$ we have $L\prec_B\beta$ and there is a $C \in N\cup U$, $L' \in C$, and a grounding $\sigma$ such that $L' \sigma = L$.
  \end{enumerate}
\end{definition}

To show soundness of SCL, we first show soundness of the initial state. Then, we show that all SCL rule applications preserve soundness, which shows soundness of the overall calculus starting from the initial state.

% SOUNDNESS_MARKER
\begin{lemma}[Soundness of the initial state]
  \label{lemma:soundness_initial}
  The initial state $(\epsilon; N; \emptyset; \beta; 0; \top)$ is sound.
\end{lemma}
\begin{proof}
  Criteria 1--3 and 6 are trivially satisfied by $\Gamma = \epsilon$. Furthermore, $N \models \emptyset$, fulfilling criterion 4. Lastly, criterion 5 is trivially fulfilled for $D = \top$.
\end{proof}

%Note that an initial state $(\epsilon, N, \emptyset, \beta, 0, \top)$ is sound.
%A rule is \emph{sound} if it maps sound states to sound states.

\begin{theorem}[Soundness of \SL{}]\label{theo:scl:soundness}
  All \SL{} rules preserve soundness, i.e. they map a sound state onto a sound state.
\end{theorem}
\begin{proof}
  As the hypothesis, assume that a state $(\Gamma; N; U; \beta; k; D)$ is sound. We show that any application of a rule results again in a sound state.

%  \TODO{hier einmal klarmachen, dass soundness die hypo ist}

    \medskip
    \noindent
  \rulename{Decide}. Assume Decide is applicable to $(\Gamma; N; U; \beta; k; D)$, yielding a resulting state $(\Gamma, L\sigma^{k+1}; N; U; \beta; k+1; D)$. Then there is a $L\in C$ for $C \in N \cup U$, $L\sigma$ is ground and undefined in $\Gamma$, and $L\sigma \prec_B \beta$. Also, there can be no active conflict, i.e. $D=\top$.
  \begin{itemize}
    \item[1, 3] By the precondition, $L\sigma$ is undefined in $\Gamma$ (3). Hence, adding $L\sigma$ does not make $\Gamma$ inconsistent (1).
    \item[2, 4] Trivially fulfilled by hypothesis.
    \item[5] Since $D = \top$, the rule is trivially satisfied.
    \item[6] For all literals $L'\sigma' \in \Gamma$, this holds by hypothesis. For $L\sigma$ this follows directly from the preconditions of the rule.
  \end{itemize}

  %RULE-BEGINS: Propagate
  \noindent
  \rulename{Propagate}. Assume Propagate is applicable to $(\Gamma; N; U; \beta; k; D)$, yielding a resulting state $(\Gamma, L\sigma^{(C_0 \lor L){\delta}\cdot\sigma};N;U;\beta;k;D)$.
  Then, there is a $C\lor L\in (N\cup U)$ such that $C = C_0 \lor C_1$, $C_1\sigma = L\sigma \lor \dots \lor L\sigma$,
  $C_0\sigma$ does not contain $L\sigma$, $\delta$ is the mgu of the literals in $C_1$ and $L$, $(C\lor L)\sigma$ is ground, $(C\lor L)\sigma \prec_B \{\beta\}$,
  $C_0\sigma$ is false under $\Gamma$, and $L\sigma$ is undefined in $\Gamma$.
  Also, there can be no active conflict, i.e. $D=\top$.
  \begin{itemize}
    \item[1, 3] By the precondition, $L\sigma$ is undefined in $\Gamma$ (3). Hence, adding $L\sigma^{(C_0 \lor L){\delta}\cdot\sigma}$ does not make $\Gamma$ inconsistent (1).
    \item[2] Consider any decomposition $\Gamma, L\sigma^{(C_0 \lor L){\delta}\cdot\sigma} = \Gamma_1, L'\sigma'^{C'_0\lor L' \cdot \sigma'}, \Gamma_2$.
    In the case of $L'\sigma \not= L\sigma$, we can apply the hypothesis for the state $(\Gamma; N; U; \beta; k; D)$.
    Hence, only the case  $\Gamma_1 = \Gamma$, $L'\sigma' = L\sigma$, and $C'_0\sigma = C_0\sigma$ is left to prove.

    First, note that $C_0\sigma$ is false under $\Gamma_1 = \Gamma$ by the preconditions. Also, $L\sigma$ must be undefined in $\Gamma$ by the preconditions. Lastly, it needs to be shown that $N \cup U \models (C_0 \lor L)\delta$. Clearly, since $C\lor L \in (N \cup U)$, it holds that $N \cup U \models C \lor L$.
    %$(C_0 \lor C_1 \lor L)\sigma$ is an instance of $C\lor L$.
    Since $C=C_0 \lor C_1$ and $C_1\sigma = L\sigma \lor \dots \lor L\sigma$ it follows from the soundness of first-order factorization that $C \models (C_0 \lor L)$ and by this $N\cup U \models C_0 \lor L$.
    \item[4] Follows trivially from the induction hypothesis.
    \item[5] Since $D = \top$, this rule is trivially satisfied.
    \item[6] For all literals $L'\sigma' \in \Gamma$, this holds by hypothesis. For $L\sigma$, consider the precondition that $(C \lor L)\sigma \prec_B \{\beta\}$. By the definition of the multiset extension of $\prec_B$, it follows that $L\sigma \prec_B \beta$ must hold as well.
  \end{itemize}

  \noindent
  \rulename{Conflict}. Assume Conflict is applicable to $(\Gamma; N; U; \beta; k; D)$, yielding a resulting state $(\Gamma;N;U;\beta;k;C\cdot\sigma)$.
  Then, there is a $C\in (N\cup U)$ such that $C\sigma$ is false in $\Gamma$ for a grounding $\sigma$.

  \begin{itemize}
    \item[1-3] Trivially fulfilled by hypothesis, as the trail $\Gamma$ is not modified.
    \item[4] Follows trivially from the induction hypothesis, as neither $N$ nor $U$ are modified.
    \item[5] It holds that $D = C\cdot\sigma$. By the preconditions of Conflict, $C\sigma$ must be false under $\Gamma$. Furthermore, since $C \in (N \cup U)$ it holds that $N\cup U \models C$. Since $N \models U$ by soundness (4), it also holds that $N \models C$. Lastly, it remains to show that $\mGndB(N) \models C\sigma$. By soundness (6), we know that for all literals $L\mu \in \Gamma$ it holds that $L\mu \prec_B \beta$. Since $C\sigma$ is false in $\Gamma$, it must hold that all literals in $C\sigma$ are also $\prec_B \beta$. Combined with $N \models C$, this yields that $\mGndB(N) \models C\sigma$.
    \item[6] Fulfilled by the hypothesis, since no literal is added to $\Gamma$.
  \end{itemize}

   %RULE-BEGINS: Skip
   \noindent
   \rulename{Skip}. Assume Skip is applicable to $(\Gamma = \Gamma', L; N; U; \beta; k; D\cdot\sigma)$, yielding a resulting state $(\Gamma';N;U;\beta;k-i;{D\cdot\sigma})$.
   By the preconditions of skip, it must hold that $\comp(L)$ does not occur in $D\sigma$, and if $L$ is a decision literal then $i = 1$ else $i = 0$.
   \begin{itemize}
    \item[1-3, 6] Directly fulfilled by hypothesis, as all prefixes of $\Gamma$ still fulfil all properties. In particular, this holds for the prefix $\Gamma'$ of $\Gamma$.
    \item[4] Follows trivially from the induction hypothesis, as $U$ is not modified.
    \item[5] After the application of Skip, ${D\cdot \sigma}$ is the current conflict. Since $D$ is not modified, $N \models D$ and $\mGndB(N) \models D\sigma$ by hypothesis. It is left to show that $D\sigma$ is false under the resulting $\Gamma'$, given the assumption that $D\sigma$ is false under $\Gamma$.
    However, since $\comp(L) \not\in D\sigma$, this is trivially fulfilled, as the removal of $\comp(L)$ from the trail $\Gamma$ cannot make $D\sigma$ undefined. Hence, $D\sigma$ must be false under $\Gamma'$ as well.
  \end{itemize}

   %RULE-BEGINS: Factorize
   \noindent
   \rulename{Factorize}. Assume Factorize is applicable to $(\Gamma; N; U; \beta; k; (D\lor L \lor L')\cdot\sigma)$, yielding a resulting state $(\Gamma;N;U;\beta;k;(D\lor L)\eta\cdot\sigma)$.
    Then, $L\sigma = L'\sigma$ and $\eta=\mMGU(L,L')$.

    \begin{itemize}
      \item[1-3, 6] Trivially fulfilled by hypothesis, as the trail $\Gamma$ is not modified.
      \item[4] Follows trivially from the induction hypothesis, as $U$ is not modified.
      \item[5] After the application of Factorize, ${(D\lor L)\eta\cdot\sigma}$ is the current
      conflict. By the hypothesis $N \models (D\lor L \lor L')$. From the preconditions of Factorize, $L\sigma = L'\sigma$ and $\eta=\mMGU(L,L')$. Thus, $(D\lor L \lor L')\eta$ is an instance of $(D\lor L \lor L')$ and $N \models (D\lor L \lor L')\eta$. Since $L\eta = L'\eta$, $(D\lor L \lor L')\eta \models (D \lor L')\eta$. Thus, $N \models (D\lor L)\eta$. By the preconditions, $\mGndB(N) \models \mGndB((L\lor L\lor L')\sigma)$. Hence, $(D\lor L\lor L')\sigma \prec_B \{\beta\}$. Thus, $(D\lor L)\eta\sigma = (D\lor L)\sigma \prec_B \{\beta\}$. From this, it follows that $\mGndB(N) \models \mGndB((D \lor L)\sigma)$.

      Furthermore, $(D\lor L)\eta\sigma$ is false under $\Gamma$, since $(D\lor L)\eta\sigma = (D\lor L)\sigma$ by the definition of an mgu, and $(D\lor L\lor L')\sigma$ is already false under $\Gamma$.
    \end{itemize}

    %RULE-BEGINS: Resolve
    \noindent
    \rulename{Resolve}. Assume the rule Resolve is applicable to an SCL state of the shape $(\Gamma = \Gamma', L\delta^{(C\lor L)\cdot\delta};N;U;\beta;k;{(D\lor L')\cdot\sigma})$, yielding a resulting state $(\Gamma;N;U;\beta;k;{(D\lor C)\eta\cdot\sigma\delta})$.
      By the preconditions of Resolve, it holds that $L\delta = \comp(L'\sigma)$ and $\eta=\mMGU(L,\comp(L'))$.

      \begin{itemize}
        \item[1-3, 6] Trivially fulfilled by hypothesis, as the trail $\Gamma$ is not modified.
        \item[4] Follows trivially from the induction hypothesis, as $U$ is not modified.
        \item[5] After the application of Resolve, ${(D\lor C)\eta\cdot\sigma\delta}$ is the current
        conflict.

        By the hypothesis, $(D \lor L')\sigma$ is false under $\Gamma$. In particular, $D\sigma$ is false under $\Gamma$. By soundness (2), we know that $C\delta$ must be false under $\Gamma$ as well. Hence, $(D \lor C)\eta\sigma\delta$ is false under $\Gamma$.

        Furthermore, by the hypothesis, $N \models (D\lor L')$. Since $(D \lor L')\eta$ is an instance of $(D \lor L')$, it holds that $N \models (D \lor L')\eta$. Furthermore, by soundness (2) we know that $N\cup U \models (C \lor L)$ and by
        soundness (4) this implies that $N \models (C \lor L)$. With similar argumentation, also $N \models (C\lor L)\eta$. By the soundness of resolution, this implies $N \models (D\lor C)\eta$.

        Lastly, since $(D\lor L')\sigma$ is false in $\Gamma$, all occuring literals in $\{(D \lor L')\sigma\} \prec_B \{\beta\}$. With similar argumentation, $\{(C \lor L)\delta\} \prec_B \{\beta\}$. Hence, in particular, $(D\lor C)\eta\sigma\delta \prec_B \{\beta\}$ and, thus, $\mGndB(N) \models \mGndB((D\lor C)\eta\sigma\delta)$.
      \end{itemize}

      \noindent
      \rulename{Backtrack}. Assume the rule Backtrack is applicable to a SCL state of shape $(\Gamma = \Gamma_0,K,\Gamma_1;N;U;\beta;k;{(D\lor L)\cdot\sigma})$, yielding the resulting SCL state $(\Gamma_0,K;N;U\cup\{D\lor L\};\beta;k';\top)$.

      \begin{itemize}
       \item[1-3, 6] Directly fulfilled by hypothesis, as all prefixes of $\Gamma$ still fulfil all properties. In particular, this holds for the prefix $\Gamma_0,K$ of $\Gamma$. %\TODO{in particular gamma' = gamma0 in this rule. or whatever, syntax gleich}
       \item[4] By the hypothesis, we know that $N \models U$. By soundness (5) we know that $N \models (D\lor L)$. Overall, $N \models U \cup \{D \lor L\}$
       \item[5] Since after an application of Backtrack the conflict is resolved, i.e. $D= \top$, the rules are trivially satisfied.
      \end{itemize}

\end{proof}

\begin{corollary}\label{corr:scl:full-soundness}
  The rules of \SL{} are sound, hence \SL{} starting with an initial state is sound.
\end{corollary}

\begin{proof}
  Follows by induction over the size of the run. The base case is handled by Lemma \ref{lemma:soundness_initial}, the induction step is contained in Theorem \ref{theo:scl:soundness}.
\end{proof}

Next we introduce reasonable and regular runs. Note that the following definitions are now changed. The new definitions have been adapted from \cite{BrombergerEtAl2020arxiv} to fit the classical SCL calculus.
As an overall goal, we will show that regular runs always generate non-redundant clauses

\begin{definition}[Reasonable Runs]\label{defi:reasonable}
  A sequence of SCL rule applications is called a \emph{reasonable run} if the rule Decide does not enable
  an immediate application of rule Conflict.
\end{definition}

\begin{definition}[Regular Runs]\label{defi:regular}
A sequence of SCL rule applications is called
a \emph{regular run} if it is a reasonable run and the rule Conflict has precedence over all
other rules.
%Furthermore, after every conflict, Resolve resolves away at least the rightmost literal
%from the trail
\end{definition}

\begin{theorem}[Correct Termination] \label{theo:scl:corterm}
  If in a regular run
  no rules are applicable to a state $(\Gamma;N;U;\beta;k;D)$ then either $D=\bot$ and $N$ is
  unsatisfiable or $D = \top$ and $\mGnd(N)^{\prec_B \beta}$ is satisfiable and $\Gamma\models \mGnd(N)^{\prec_B \beta}$.
\end{theorem}
\begin{proof}

\begin{extendedonly}
  Consider a state $(\Gamma;N;U;\beta;k;D)$. Then, $D$ can have one of the following shapes:

  \medskip
  \noindent \emph{(Case $D = \top$)} If $D = \top$, then there is no active conflict.
  Assume there are no undefined ground literals $L \prec_B \beta$ for $L \in C$, $C \in N\cup U$ in $\Gamma$. Now, either  $\Gamma \models \mGndB(N)$ and thus $\Gamma$ is already a partial model for $N$ w.r.t. $\prec_B$ and $\beta$. Otherwise, if $\Gamma \not\models \mGndB(N)$ but all literals are defined, there must be a false clause $C \in \mGndB(N)$ which can be chosen as a Conflict instance.

  If there is at least one undefined ground literal $L \prec_B \beta$ occuring in $N \cup U$, one of the trail building rules Propagate, Decide, or Conflict are applicable. Decide on the undefined ground literal $L$ is, by the preconditions of the rule, in such a case always possible. The application of Decide can, however, be restricted by reasonability or regularity.

  If Decide on $L$ is not applicable by reasonability,
  then $\Gamma, {L}^{k+1}$ must lead to a direct application of Conflict. Thus, there is a clause $D \in N \cup U$ such that $D\sigma$ is false under $\Gamma, L^{k+1}$. If $D\sigma$ is already false under $\Gamma$, then Conflict is applicable. Otherwise, $D$ has the shape $D_0 \lor D_1$ where $D_0$ is false under $\Gamma$, and $D_1\sigma = \comp(L) \lor \dots \lor \comp(L)$. Since $D_0$ is false under $\Gamma$, also $D_0 \prec_B \{\beta\}$ and since $L \prec_B \beta$ it holds that $D_0 \lor D_1 \prec_B \{\beta\}$ by the definition of our multiset extension. Hence, Propagate can be applied.

  If Decide is not applicable by regularity, Conflict must be applicable, since regularity only priorizes the Conflict rule application.

  \medskip
  \noindent \emph{(Case $D = C\cdot \sigma$)} If $D = C\cdot \sigma$, then there is an active conflict which needs to be resolved. In this case, one of the rules Resolve, Skip, Factorize or Backtrack are applicable.

  First, consider the case of $\Gamma = \varepsilon$. By soundness, $C\sigma$ must be false under $\Gamma$. However, the only false clause under $\varepsilon$ is $\bot$. In this case, $D = \bot$ and by soundness, $N \models \bot$. Hence, $N$ is unsatisfiable.
  In the other case, there is at least one literal on the trail. We split $\Gamma = \Gamma', L$ and distinguish the shape of $L$:
  \begin{itemize}
    \item Consider the case that $L$ was propagated, i.e. is of shape $L^{C\cdot \delta}$ for a clause $C \in N \cup U$.
    Then, either Resolve or Skip are applicable. In the case that $\comp(L)$ occurs in $C\sigma$, Resolve is applicable. If $\comp(L) \not\in C\sigma$, Skip is applicable.
    \item Consider the case that $L$ is a decision literal, i.e. is of shape $L^i$ for a numerical level $i$. Then, one of the rules Skip, Backtrack or Factorize are applicable.

    If $\comp(L)$ does not occur in $C\sigma$, then Skip can be applied. Backtrack can be applied in all other cases if $C = (C' \lor \comp(L))$, where $C'$ is of level $i' < k$. Note that for Backtrack there must be a level $j$ that is backtracked to. This level $j$ always exists if all other preconditions are met. Hence, if Skip is not applicable, $C$ is of the shape $C' \lor \comp(L)$. If $C'$ is of level $k$, then Factorize can be applied instead, as $C'$ must contain another instance of $\comp(L)$. Otherwise, $C'$ is of level $i' < k$ and Backtrack can be applied.
  \end{itemize}

\end{extendedonly}

\begin{notinextended}
  For a state  $(\Gamma;N;U;\beta;k;D)$ where
  $D\not\in\{\top, \bot\}$, one of the rules Resolve, Skip, Factorize or Backtrack
  is applicable. If the top level literal is a propagated literal
  then either Resolve or Skip are applicable. If the top level
  literal is a decision then one of the rules Skip, Backtrack, or Factorize is applicable.
  In the case $D=\top$ and Decide is not applicable by regularity, Propagate can always be applied instead.
  If $D=\top$ and all Propagate, Decide, and
  Conflict are not applicable it means that there are no undefined ground literals $L \prec_B \beta$
  in $\Gamma$, so $\Gamma\models \mGnd(N)^{\prec_B \beta}$.
\end{notinextended}

\end{proof}

%\begin{corollary}[Decide cannot create Conflicts]\label{corol:no-decide-conflict}
%  Any application of Decide in an SCL regular run from starting state $(\epsilon;N;\emptyset;\beta;0;\top)$ does %not create a conflict.
%\end{corollary}
%\begin{proof}
%  A regular run is also reasonable, and Decide cannot be applied to create a conflict in a reasonable run.
%\end{proof}

\begin{lemma}[Resolve in regular runs]\label{lemma:resolve_in_regular}
  Consider the derivation of a conflict state \newline$(\Gamma, L;N;U;\beta;k;\top) \Rightarrow_{\text{Conflict}} (\Gamma, L;N;U;\beta;k;D)$.
  In a regular run, during conflict resolution $L$ is not a decision literal and at least the literal $L$ is resolved.
\end{lemma}
\begin{proof}
  There are three ways how a reasonable run can reach a state $(\Gamma, L;N;U;\beta;k;\top)$: either by an application of the rule Decide, Backtrack, or Propagate.
  In a reasonable run, if the rule Decide has produced the SCL state $(\Gamma, L;N;U;\beta;k;\top)$, the rule Conflict is not immediately applicable.
  In case the rule Backtrack produced the state  $(\Gamma, L;N;U;\beta;k;\top)$ there is the sequence of rule applications
  \begin{align*}
    &~(\Gamma, L, L',\Gamma_1,K^{k+1},\Gamma_2,\comp(L''\sigma)^{k'};N;U';\beta;k';(D\lor L'')\cdot\sigma)
    \\ \Rightarrow_\text{SCL}^\text{Backtrack} &~(\Gamma, L;N;U'\cup(D\lor L'');\beta;k;\top)
  \end{align*}
  %\TODO{check gamma and gamma' 0 l}
  Then, by the definition of Backtrack, the newly learned clause $(D\lor L'')$ cannot be false with respect to $\Gamma, L$.
  This also means that any clause in $N \cup U'$ can also not be false with respect to $\Gamma, L$:
  (i)~either it would have been present the last time when the run visited a state with trail $\Gamma, L$ and in a reasonable run this means that the rule Conflict should have been applied at that point and prevented any further exploration of the trail prefix $\Gamma, L$ or (ii)~it was learned afterwards via the rule Backtrack in a state with the trail prefix $\Gamma, L$; however, this would contradict the previously proven fact that Backtrack always jumps to a trail prefix where the newly learned clause cannot be false anymore.
  Hence, in a reasonable run the rules Decide and Backtrack are never applied right before the rule Conflict.
  In summary, $L$ must be a propagated literal if Conflict is applicable to $(\Gamma, L;N;U;\beta;k;\top)$ from a regular run.

  Backtrack is not directly applicable to $(\Gamma, L;N;U;\beta;k;D)$, as it requires $L$ to be a decision literal.
  Furthermore, $L$ must occur in the conflict clause $D$. Otherwise, Conflict could have been applied
  earlier to $(\Gamma;N;U;\beta;k;\top)$, contradicting regularity. Hence, Skip is not applicable to our state.
  Overall, only Factorize and Resolve can possibly be applied to our state. After an application of Factorize, the two invariants still hold: First, the trail is not modified. Second, $L$ must still occur in the conflict clause $D$, as Factorize cannot remove all instances of $L$ from $D$. Hence, Factorize cannot enable any of the rules Skip or Backtrack. Following from that, at least one application of Resolve must take place in conflict resolution.
\end{proof}

\begin{definition}[State Induced Ordering] \label{def:state_ordering}
  Let $(L_1, L_2,\ldots,L_n;N;U;\beta;k;D)$ be a sound state of {\SL}. The trail
  induces a total well-founded strict order on the defined literals by\newline
\centerline{$L_1\prec_\Gamma\comp(L_1)\prec_\Gamma L_2\prec_\Gamma\comp(L_2)\prec_\Gamma\cdots\prec_\Gamma L_n\prec_\Gamma\comp(L_n).$}
  We extend $\prec_\Gamma$ to a strict total order on all literals where
  all undefined literals are larger than $\comp(L_n)$.
  We also extend $\prec_\Gamma$ to a strict total order on ground clauses by
  multiset extension and also on multisets of ground clauses and overload $\prec_\Gamma$
  for all these cases.
%  In the following we sometime identify the order $\prec_\Gamma$ with $\Gamma$
  With $\preceq_\Gamma$ we denote  the reflexive closure of $\prec_\Gamma$.
\end{definition}

\begin{theorem}[Learned Clauses in Regular Runs]\label{theo:SnnRed}
  Let $(\Gamma;N;U;\beta;k;C_0\cdot\sigma_0)$ be the state
  resulting from the application of Conflict in
  a regular run and let $C$ be
  the clause learned at the end of the conflict resolution,
  then $C$ is not redundant with respect to $N\cup U$ and $\prec_\Gamma$.
\end{theorem}
\begin{proof}
\begin{notinmaster}
    Consider the following fragment of a derivation learning a clause:\newline
   \centerline{$\Rightarrow^{\text{Conflict}}_{\text{SCL}}(\Gamma;N;U;\beta;k;C_0\cdot\sigma_0) \Rightarrow^{\{\text{Skip, Fact., Res.}\}^*}_{\text{SCL}} (\Gamma';N;U;\beta;k;C\cdot\sigma)\Rightarrow^{\text{Backtrack}}_{\text{SCL}}.$}
\end{notinmaster}
 \begin{masteronly}
    Consider the following fragment of a derivation learning a clause:
  \begin{alignat*}{3}
    &\Rightarrow^{\text{Conflict}}_{\text{SCL}} && (\Gamma;N;U;\beta;k;C_0\cdot\sigma_0) &&  \\
    &\Rightarrow^{\{\text{Skip, Fact., Res.}\}^*}_{\text{SCL}} && (\Gamma';N;U;\beta;k;C\cdot\sigma)
    \\ & \Rightarrow^{\text{Backtrack}}_{\text{SCL}} && &&
 \end{alignat*}
\end{masteronly}

  By soundness $N\cup U \models C$ and $C\sigma$
  is false under both $\Gamma$ and $\Gamma'$.
  We prove that $C\sigma$ is non-redundant to $N \cup U$ with respect to $\prec_{\Gamma}$.

  Assume there is an $S \subseteq \mGnd(N\cup U)^{\preceq_{\Gamma} C\sigma}$
  \st $S\models C\sigma$. There must be a clause $D\in S$
  false under $\Gamma$, since all clauses in $S$ have a defined truth value
  (as all undefined literals are greater in $\prec_\Gamma$ than all
  defined literals)
  and if $\Gamma \models S$ then $\Gamma \models C\sigma$ by transitivity of entailment, a contradiction.

  %$S \preceq_\Gamma \{C\sigma\}$ and $C\sigma\not\in S$.

  By regularity, $\Gamma$ must be of the shape $\Gamma = \Gamma'', L\delta^{C\lor L\cdot\delta}$, since no application of Decide can lead to an application of the rule Conflict. Thus, the last applied rule must have been Propagate. Furthermore, by Lemma \ref{lemma:resolve_in_regular}, Resolve must have resolved at least the rightmost literal $L\delta$ from $\Gamma$. Thus, $L\delta \not\in C\sigma$ and $\comp(L\delta) \not\in C\sigma$.

  Since $D\prec_{\Gamma}C\sigma$, neither $L\delta$ nor $\comp(L\delta)$ may occur in $D$. However, this is a contradiction, since $D$ is then already false under $\Gamma''$ and, thus, must have been chosen as a Conflict instance earlier in a regular run.
  %In constrast, note that $C_0\sigma_0$ must contain $L\delta$ or $\comp(L\delta)$, since otherwise $C_0\sigma_0$ was already defined under $\Gamma''$ and Conflict could have been applied instead of Propagate, a contradiction to regularity.
\end{proof}

%\begin{remark} \label{rem:invariant_ordering}
%\TODO{nicht kursiv}
Of course, in a regular run, the ordering of foreground literals on the trail will change, i.e., the ordering
of Definition~\ref{def:state_ordering} will change as well. Thus the non-redundancy property of Lemma~\ref{theo:SnnRed}
reflects the situation at the time of creation of the learned clause. A non-redundancy property holding for
an overall run must be invariant against changes on the ordering.
However, the ordering of Definition~\ref{def:state_ordering} also entails
a fixed subset ordering $\prec_\subseteq$ that is invariant against changes on the overall ordering.
This means that our dynamic ordering entails non-redundancy criteria based on subset relations including forward subsumption. % redundancy \TODO{redundancy vs subsumptio}.
From an implementation perspective,
this means that learned clauses need
not to be tested for forward redundancy. Current resolution or superposition based provers spent a reasonable portion
of their time in testing forward redundancy of newly generated clauses. In addition, also tests for backward reduction can be
restricted knowing that learned clauses are not redundant.
%\end{remark}

\begin{theorem}[FOL Non-Redundancy is Undecidable] \label{theo:cslfolnonred}
  Deciding non-redundancy of a first-order clause $C$ with respect to a finite first-order clause set $N^{\preceq C}$ is undecidable.
\end{theorem}
\begin{proof}
  We reduce the problem to unsatisfiability of an arbitrary first-order clause set $N$.
  Let $N = \{C_1,\ldots,C_n\}$ be an arbitrary, finite first-order clause set. We consider an LPO ordering $\prec_{\text{LPO}}$.
  Next, we add a fresh predicate $P$ of arity zero, where $P$ is $\prec_{\text{LPO}}$ larger than any clause in $N$. Now, in the finite first-order clause set $N \cup \{P\}$ the clause $P$
  is redundant iff $N$ is unsatisfiable.
\end{proof}

Another aspect of our definition of regular runs that it does not require exhaustive propagation.
This is important for examples where propagation may already refute a problem, but particular learning
steps provide an exponentially shorter proof.

\begin{example}[Delayed Propagation]
  Consider the clause set
  \[\begin{array}{rl}
  N = \{(1) & P(0,0,0,0)\\
(2) & \lnot P(x_1,x_2,x_3,0) \lor P(x_1,x_2,x_3,1)\\
(3) & \lnot P(x_1,x_2,0,1) \lor P(x_1,x_2,1,0)\\
(4) & \lnot P(x_1,0,1,1) \lor P(x_1,1,0,0)\\
(5) & \lnot P(0,1,1,1) \lor P(1,0,0,0)\\
  (6) & \lnot P(1,1,1,1)\}\\
  \end{array}\]
  where exhaustive propagation enumerates the four bit counter from $P(0,0,0,0)$ to $P(1,1,1,1)$ and
  this way refutes the clause set. A proof of quadratic length~\cite{PerezVoronkov08} can be obtained as follows
  \[
  \renewcommand{\arraystretch}{1.3}
   \begin{array}[]{ll}
   & (\varepsilon; N; \emptyset; \beta; 0; \top) \\
   \mrulename{Propagate}     & (P(0,0,0,0)^{(1) \cdot \{\}}; N; \emptyset; \beta; 0; \top) \\
   \mrulename{Decide   }     & (P(0,0,0,0)^{(1)\cdot \{\}}, \lnot P(0,0,1,0)^1; N; \emptyset; \beta; 1; \top) \\
   \mrulename{Propagate}     & (P(0,0,0,0)^{(1)\cdot \{\}}, \lnot P(0,0,1,0)^1, \lnot P(0,0,0,1)^{(3)\cdot\{x_1\mapsto 0, x_2\mapsto 0\}}; N; \emptyset; \beta; 1; \top) \\
   \mrulename{Conflict}      & (P(0,0,0,0)^{(1)\cdot \{\}}, \lnot P(0,0,1,0)^1, \lnot P(0,0,0,1)^{(3)\cdot\{x_1\mapsto 0, x_2\mapsto 0\}}; \\
   & \quad N; \emptyset; \beta; 1; (2) \cdot \{x_1\mapsto 0, x_2\mapsto 0, x_3\mapsto 0\}) \\
   \mrulename{Resolve,Skip}    & (P(0,0,0,0)^{(1)\cdot \{\}}, \lnot P(0,0,1,0)^1;\\
   & \quad N; \emptyset; \beta; 1; \lnot P(x_1,x_2,0,0) \lor P(x_1,x_2,1,0) \cdot \{x_1\mapsto 0, x_2\mapsto 0\}) \\
   \mrulename{Backtrack}       & (P(0,0,0,0)^{(1)\cdot \{\}}; N\cup\{(7)~\lnot P(x_1,x_2,0,0) \lor P(x_1,x_2,1,0)\};  \emptyset; \beta; 0;   \top) \\
  \end{array}
  \]
  where we accelerated the counter on the first to bit. Using the same technique we can eventually learn the clause $\lnot P(x_1,0,0,0) \lor P(x_1,1,1,1)$
  accelerating the first three bits. The technique requires lazy propagation. In order to obtain a resolution step between two of the two literal
  clauses we decide the positive literal of the clause where we want to resolve with the negative literal and we make sure the positive literal we want to resolve
  with triggered a propagation. This is supported by a regular run.
\end{example}

\begin{theorem}[BS Non-Redundancy is NEXPTIME-complete] \label{theo:cslnonrednexptime}
  Deciding non-redundancy of a BS clause $C$ with respect to a finite BS clause set $N^{\preceq C}$ is NEXPTIME-complete.
\end{theorem}
\begin{proof}
  We only show hardness, because containment of the problem in NEXPTIME is obvious.
  To this end, let $N = \{C_1,\ldots,C_n\}$ be an arbitrary, finite BS clause set. We consider an LPO ordering $\prec_{\text{LPO}}$.
  Next, we add a fresh predicate $P$ of arity zero, where $P$ is $\prec_{\text{LPO}}$ larger than any clause in $N$. Now, in the finite BS clause set $N \cup \{P\}$ the clause $P$ is redundant iff $N$ is unsatisfiable.
\end{proof}

\begin{theorem}[Termination] \label{theo:finite-termination}
 Any regular run of $\Rightarrow_{\mSL}$ terminates. %\TODO{hand-wavy}
\end{theorem}
\begin{proof}
  Any infinite run learns infinitely many clauses.
  Firstly, for a regular run, by Theorem~\ref{theo:SnnRed},
  all learned clauses are non-redundant. Those clauses are also non-redundant under the fixed subset ordering $\prec_\subseteq$, which is well-founded.
  Due to the restriction of all clauses to be smaller than $\{\beta\}$, the overall number of non-redundant ground clauses is finite. So there is no infinite regular run.
 \end{proof}

\begin{theorem}[SCL Refutational Completeness] \label{theo:scl:refcomplete}
  If $N$ is unsatisfiable, such that some finite $N'\subseteq\mGnd(N)$ is unsatisfiable and $\beta$ is $\prec_B$ larger than all literals in $N'$ then any regular run
  from $(\epsilon; N; \emptyset; \beta; 0; \top)$ of SCL derives $\bot$.
\end{theorem}
\begin{proof}
  By Theorem~\ref{theo:finite-termination} and Theorem~\ref{theo:scl:corterm}.
\end{proof}

Of course, for a given, unsatisfiable clause set $N$, a sufficiently large $\beta$ cannot be effectively computed in advance, in general.
Then the calculus is extended with a different rule that increases $\beta$. The below rule grow has to be applied in a fair way together
with the other SCL calculus resulting then in a semi-decision procedure.

\bigskip
\shortrules{Grow}
{$(\Gamma;N;U;\beta;k;\top)$}
{$(\epsilon;N;U;\beta';0;\top)$}
{provided $\beta \prec_B \beta'$}{\mSL}{12}

\begin{example}
  For example, consider a SCL run on the following clause set:
  \begin{center}
    $N = \left\{\begin{array}{rl}
      (1) & \neg P(x) \lor P(g(x)) \\
      (2) & P(a)\\
      (3) & \neg P(g(g(a)))\\
    \end{array}\right\}$
    \end{center}

    First, we choose $\prec_B$ as a KBO, where all symbols have weight one and with precedence $a\prec g\prec P$. In the beginning of the run, we set $\beta = P(g(g(a)))$.
%    Now perform an SCL regular run on $N$, leading to a stuck state,
%    apply Grow appropriately and eventually refute $N$.

  \renewcommand{\arraystretch}{1.2}
  \[
   \begin{array}[]{ll}
   & (\varepsilon; N; \emptyset; \beta; 0; \top) \\
   \mrulename{Propagate}     & (P(a)^{(2) \cdot \{\}}; N; \emptyset; \beta; 0; \top) \\
   \mrulename{Propagate}     & (P(a)^{(2) \cdot \{\}}, P(g(a))^{(1) \cdot \{x \mapsto a\}}; N; \emptyset; \beta; 0; \top) \\
   \mrulename{Grow}          & (\varepsilon; N; \emptyset; \beta' = P(g(g(g(a)))); 0; \top) \\
   \mrulename{Propagate}     & (P(a)^{(2) \cdot \{\}}; N; \emptyset; \beta; 0; \top) \\
   \mrulename{Propagate}     & (P(a)^{(2) \cdot \{\}}, P(g(a))^{(1) \cdot \{x \mapsto a\}}; N; \emptyset; \beta; 0; \top) \\
   \mrulename{Propagate}     & (P(a)^{(2) \cdot \{\}}, P(g(a))^{(1) \cdot \{x \mapsto a\}},  P(g(g(a)))^{(1) \cdot \{x \mapsto g(a)\}}; N; \emptyset; \beta'; 0; \top) \\
   \mrulename{Conflict}      & (P(a)^{(2) \cdot \{\}}, P(g(a))^{(1) \cdot \{x \mapsto a\}},  P(g(g(a)))^{(1) \cdot \{x \mapsto g(a)\}}; N; \emptyset; \beta'; 0; (3) \cdot \{\}) \\
   \mrulename{Resolve,Skip}       & (P(a)^{(2) \cdot \{\}}, P(g(a))^{(1) \cdot \{x \mapsto a\}}; N; \emptyset; \beta'; 0; \neg P(g(a)) \cdot \{x \mapsto g(a)\}) \\
   \mrulename{Resolve,Skip}       & (P(a)^{(2) \cdot \{\}}; N; \emptyset; \beta'; 0; \neg P(a) \cdot \{x \mapsto g(a), x' \mapsto a\}) \\
   \mrulename{Resolve}       & (P(a)^{(2) \cdot \{\}}; N; \emptyset; \beta'; 0; \bot \cdot \{x \mapsto g(a), x' \mapsto a\}) \\
  \end{array}
  \]

  Note that after the second Propagate, no other rule except Grow is applicable to our state. In particular, note that after this step, $ \Gamma = \{ P(a), P(g(a)) \} \models \mGndB(N)$. Hence, no further literal can be added to the trail and no conflict can be detected. Thus, the only option here is to use Grow to consider more literals. Only with these additional literals, a refutation is eventually found.
\end{example}

\begin{theorem}[SCL decides the BS fragment] \label{theo:scl:decidesbs}
  SCL restricted to regular runs decides satisfiability of a BS clause set if $\beta$ is set appropriately.
\end{theorem}
\begin{proof}
  Let $B$ be the set of constants in the BS clause set $N$. Then define $\prec_B$ and $\beta$ such that $L \sigma \prec_B \beta$ for all
  $L\in N$ and for all groundings $\sigma$ of $L$ with $\cdom(\sigma) \subseteq B$.
  Following the proof of Theorem~\ref{theo:finite-termination}, any SCL regular run will terminate on a BS clause set.
\end{proof}

%%% Local Variables:
%%% mode: latex
%%% TeX-master: "paper"
%%% End:

\section{Conclusion} \label{sec:conclusion}

The main contributions of this paper revisiting the SCL calculus for first-order are:
(i)~the incorporation of the regularity definition of \cite{BrombergerEtAl2020arxiv} to classical SCL. This enables SCL to guarantee non-redundant learning even without requiring exhaustive propagation,
(ii)~the introduction of a new trail bounding measure, which guarantees termination by limiting the size of considered literals. Still, for proper bounds, SCL remains a decision procedure for
any logic that enjoys the finite model property.
Moreover, this paper contains formal, rigorous proofs for the SCL calculus, including a full soundness proof which was omitted before \cite{FioriWeidenbach19,BrombergerEtAl2020arxiv}.
Moreover, the proofs of the original paper have been adapted and, in some cases, simplified.

This paper, hence, forms part of an overall development of the SCL calculus.
The overall goal of this process is to provide an efficient, powerful algorithm for proving first-order formulae or providing counterexamples.

% %%% Local Variables:
% %%% mode: latex
% %%% TeX-master: "paper"
% %%% End:
% 

\smallskip\noindent
%{\bf Acknowledgments:} This work was funded by DFG grant 389792660 as part of
%\href{http://perspicuous-computing.science}{TRR~248}.

\bibliographystyle{splncs04}
%\bibliography{paper}

%\input{appendix}

\end{document}